\documentclass[DIV=14]{scrartcl}
\usepackage[utf8]{inputenc}
\usepackage[english]{babel}
\usepackage[babel]{csquotes}
\usepackage{stmaryrd}
\usepackage{amssymb}
\usepackage{amsmath}
\usepackage{amsthm}
\usepackage{fancyhdr}
\usepackage{multirow}

\allowdisplaybreaks

\providecommand{\R}{\mathbb{R}}

\providecommand{\stx}{\mathbf{x}}

\providecommand{\stu}{\mathbf{u}}

\providecommand{\eps}{\epsilon}

\DeclareMathOperator{\const}{const}

\DeclareMathOperator{\re}{Re}

\DeclareMathOperator{\Ma}{Ma}

\DeclareMathAlphabet{\mathcal}{OMS}{cmsy}{m}{n}
\DeclareSymbolFont{largesymbols}{OMX}{cmex}{m}{n}

\newtheorem{lemma}{Lemma}
\newtheorem{theorem}{Theorem}
\newtheorem{remark}{Remark}
\newtheorem{corollary}{Corollary}
\newtheorem{ann}{Assumption}

\begin{document}
\title{Separation of acoustic waves in isentropic flow perturbations}

\author{%
{\sc
Christian Henke\thanks{Email: christian.henke@atlas-elektronik.com}}
}
\date{\today}

\parindent 0em

\pagestyle{fancy}
\rhead{C. Henke: Separation of acoustic waves}
\lhead{}
\chead{}
\rfoot{}
\lfoot{}
\cfoot{}
\renewcommand{\headrulewidth}{0pt}

\maketitle

{\small
{\bf Abstract:}
The present contribution investigates the mechanisms of sound generation and propagation in the case of highly-unsteady flows. 
Based on the linearisation of the isentropic Navier-Stokes equation around a new pathline-averaged base flow, 
it is demonstrated for the first time that flow perturbations of a non-uniform flow can be split into acoustic and vorticity modes, with the acoustic modes being independent of the vorticity modes.
Therefore, we can propose this acoustic perturbation as a general definition of sound.

As a consequence of the splitting result, we conclude that the present acoustic perturbation is propagated by the convective wave equation and fulfils Lighthill's acoustic analogy. 
Moreover, we can define the deviations of the Navier-Stokes equation from the convective wave equation as ``true'' sound sources.

In contrast to other authors, no assumptions on a slowly varying or irrotational flow are necessary.

Using a symmetry argument for the conservation laws, an energy conservation result and a generalisation of the sound intensity are  provided.\\

{\bf 2010} {\it Mathematics Subject Classification:}
35C20, 
35L03, 
35L65, 
76Q05.\\ 

}

\section{Introduction}
Nowadays, several manufacturers have to pay increasingly attention to the noise of highly-unsteady flows. On the one hand, there are many applications where the background noise disturbs the desired acoustic signal. On the other hand, the acoustic comfort has become an important selection criteria for consumers.  
Therefore, it is surprising that the following significant question is open until now: 
\begin{quote}
Is it possible to derive a closed system of equations for the acoustic quantities in the case of highly-unsteady flows?
\end{quote}
Hence, the general mechanisms of sound generation and propagation are also unknown.
The qualitative definition of an acoustic perturbation 
is widely accepted being 
that 
part of fluctuation which is radiating with 
a velocity that depends on the speed of sound,
whereas non-acoustic perturbations are convected by the hydrodynamic flow.
To the best of the author's knowledge, there exists no decomposition of the highly-unsteady flow variables, where the acoustic quantity  with respect to this definition satisfies a closed system of equations.
The following remark captures this situation \cite[p. 283]{Goldstein_theoreticalbasisidentifying_2009}:
\begin{quote}
It is impossible to identify the ``sources'' of sound without first defining what is actually meant by ``sound''. Unfortunately, current understanding of unsteady compressible flows, especially turbulent shear flows, is still too rudimentary to give a completely general definition of this quantity.
\end{quote}
The aim of the paper is to give an answer to the open question raised in the opening paragraph. More precisely we present a mathematical analysis where the underlying decomposition of the flow variables is necessary for the key arguments. By virtue of the fact that the acoustic fluctuations satisfy a closed equation, this pressure fluctuation is convenient for a general definition of sound.

Moreover, the acoustic pressure fluctuation may be interpreted as the solution of Lighthill's analogy and of a convective wave equation with a true sound source.

As usual, the following steps will be considered:
\begin{enumerate}
\item splitting the variables 
of the compressible Navier-Stokes equations $\rho,u,p$ 
into their base flow $(\bar{\cdot})$ and their fluctuations $(\cdot')$;
\item linearising around the base flow;
\item neglecting the dissipation mechanisms for the fluctuations;
\item separating the acoustic fluctuations from the non-acoustic perturbations.
\end{enumerate}
Motivated by the large disparity between the energy of the
hydrodynamic and acoustic variables,  
an assumption like 
\begin{equation}
\frac{|\rho'|}{|\bar{\rho}|}= O(\eps) ,\quad 
\frac{\|u'\|}{\|\bar{u}\|}=O(\eps), \quad
\frac{|p'|}{|\bar{p}|}=O(\eps), \quad 0 <\eps \ll 1,
\label{eq:mean_flow_cond}
\end{equation}
ensures a small linearisation error.
Because of the conservation form of the Navier-Stokes equations one may expect a resulting linear conservation law after step three. 
However,  conservation laws for the fluctuation components are only known in special cases, e.g. in the case of uniform mean flows w.r.t. time averaging. 
Moreover, 
an exact separation of acoustic waves from other compressibility effects is also known exclusively in this special case (cf. \cite[Splitting Theorem, p. 220]{Goldstein_Aeroacoustics_1976}).
Therefore, the key motivation of the paper was the following question:
\begin{quote}
Can we define a base flow such that the equations for the perturbations can be formulated as a conservation law and what are the consequences for the separation of the acoustic waves?
\end{quote}
It turns out that it is convenient to work with  
a base flow which is constant along the pathlines 
of the fluid: 
\begin{equation*}
\partial_t \bar{p} + u \cdot \nabla \bar{p}=0, \quad
\partial_t \bar{u} + u \cdot \nabla \bar{u}=0.
\end{equation*}
Namely, sound waves are interpreted as vibrations around the fluid particle path.   
Because of $\nabla \bar{u}=0$ in the case of 
constant base flow, 
the associated momentum equation implies $\nabla \bar{p}=0.$ 
Obviously, $\partial_t\bar{u}=0$ and $ \partial_t \bar{p}=0$ 
and the proposed splitting method seems to be 
a natural generalisation of a time averaged splitting method in the case of uniform
mean flow.

As a consequence of the acoustic closure problem, 
there is no commonly accepted set of acoustic equations available.
Several non-radiating base flow have been suggested in the literature, 
e.g. incompressible flows by the viscous/acoustic splitting method in
\cite{Hardin.Pope_Acoustic/ViscousSplittingTechnique_1994},
\cite{Shen.Sorensen_Aeroacousticmodellingof_1999}, 
\cite{Seo.Moon_PerturbedCompressibleEquations_2005} and
\cite{Seo.Moon_Linearizedperturbedcompressible_2006} 
and time averaged flows by the linearised Euler equations (LEE) and modifications 
in \cite{Bogey.Bailly.ea_Computationofflownoise_2002} and
\cite{Ewert.Schroeder_Acousticperturbationequations_2003},
but fail at least in one of the following points:
\begin{itemize}
\item The perturbation variable is not based on an exact separation of acoustic waves / is not a well defined acoustic quantity.
\item There are unknown terms which have to be approximated.
\item The approach is based on impractical assumptions, e.g. slowly varying or irrotational flow.   
\end{itemize}

If an exact separation of acoustic waves is not possible for 
low Mach number flows, the perturbation variable contains 
different length scales which are in general
responsible for grid-dependent and unstable solutions of the 
perturbation equations (cf. 
\cite{Ewert.Schroeder_Acousticperturbationequations_2003} and 
\cite{Seo.Moon_Linearizedperturbedcompressible_2006}). 

Therefore, 
the major challenge in the weakening of the closure problem is 
to ensure stability properties. We mention the main stability arguments:
\begin{itemize}
\item The underlying differential operator yields a conservation of the perturbed energy (cf. \cite{Moehring_wellposedacoustic_1999},~\cite{Ewert.Schroeder_Acousticperturbationequations_2003}). This a priori energy result is necessary for the well-posedness in the sense of Hadamard.
\item 
Motivated by the special case of uniform mean flow, where the splitting theorems 
\cite[p. 220]{Goldstein_Aeroacoustics_1976} and
\cite[Appendix A]{Ewert.Schroeder_Acousticperturbationequations_2003}
claim that vorticity fluctuations are non-acoustic perturbations, stability could be improved by requiring 
(cf. \cite{Ewert.Schroeder_Acousticperturbationequations_2003},
\cite{Seo.Moon_Linearizedperturbedcompressible_2006})
\begin{equation*}
\partial_t \left( \nabla \times u' \right)=0,
\end{equation*}
and a suitable initial condition. In other words, this stability demand claims that the sound propagation of highly-unsteady flows is realised by longitudinal waves.

Obviously, the following questions are of special interest:

\begin{quote}
Are vorticity fluctuations and their derivatives always non-acoustic perturbations for more general flows?
\end{quote}
\begin{quote}
Are acoustic perturbations of highly-unsteady flows always occurring in longitudinal waves?
\end{quote}
\end{itemize}

A second unsatisfactory approach to determine the acoustic field is the use of the so called 
acoustic analogies. 
These methods are naturally exact but incomplete theories, i.e.
they consist of an exact equation for a pressure fluctuation variable, 
where the right hand side depends also on the solution.
Hence, in order to use the exact equation as an acoustic analogy the right-hand side has to be modelled as a known sound source.
In other words, the mechanisms responsible for sound propagation are replaced 
by equivalent sources. Consequently, it is impossible to separate 
(and therefore predict) the sound propagation mechanisms from the physical 
sources of sound. Thus, for example, the famous acoustic analogy of Lighthill
(cf. \cite{Lighthill_soundgeneratedaerodynamically._1952})
is not be able to determine refraction and convection effects.

The paper is organised as follows. Section 2 introduces 
the compressible Navier-Stokes equations and the associated equations of state.
After a short review of conservation laws and perturbation equations in the presence of uniform mean flow, Section 4 concerns the pathline-averaged base flow.
The following Section is devoted to a rigorous investigation of the propagation of sound and proves some of the key results which answer the questions raised earlier: a conservation law for the first order fluctuations of a highly-unsteady flow and the Splitting Theorem, that is a closed system of equations for the acoustic quantities.
The acoustic propagation operator is applied in Section 6 to identify the true sound source in Lighthill's right-hand side.
In Section 7 we discuss the structure of hyperbolic conservation laws. For this purpose we use a symmetry argument to prove an energy conservation result. Moreover, we propose a definition of sound intensity 
which reduces to the classical definition for a medium at rest. 

\section{Preliminaries}
In the following it is assumed that all functions are sufficiently smooth such that all further operations are valid.
Let us consider the compressible Navier-Stokes equations for a fixed time $\tau>0$
\begin{align}
\frac{D \rho}{Dt} +\rho \nabla \cdot u&=0 \quad \text{in } (0,\tau) \times \R^d, \label{eq:continuity} \\
\rho \frac{D u}{Dt} +\nabla p &=\nabla \cdot \sigma \quad \text{in }  (0,\tau) \times \R^d, \label{eq:momentum} \\
\rho T\frac{Ds}{Dt}&=-\nabla \cdot q + \sigma:\nabla u \quad \text{in } (0,\tau) \times \R^d, \label{eq:energy}\\ 
(\rho,u^\top,s)^\top&=g \text{ in }\left\{ 0 \right\} \times \R^d,
\label{eq:initial_cond}
\end{align}
where $g:\R^d \to \R \times \R^d \times \R$ is a known function and
$\rho,p,s: [0,\tau) \times \R^d \to \R,$ 
$u,q:[0,\tau) \times \R^d \to \R^d,$
$\sigma:[0,\tau) \times \R^d \to \R^{d,d}$
denote the density, pressure,
entropy, velocity, heat flux and the stress tensor, respectively.
Here and in what follows, the symbol $\frac{D}{Dt}=\partial_t + u \cdot \nabla$
denotes the material derivative.
Furthermore, $I \in \R^{d,d}$ denotes the identity matrix. 
As usual,
we use Fourier's law $q=\lambda \nabla T,$ 
and for a Newtonian medium Stokes' hypothesis
\begin{equation*}
\sigma=\mu\left( \nabla u +\left(\nabla u\right)^\top -\frac{2}{3} \nabla \cdot u I \right),
\end{equation*}
with $\lambda,\mu,T: [0,\tau) \times \R^d \to \R$ the heat conductivity, dynamic viscosity and the temperature.
In order to close the system $(\ref{eq:continuity})-(\ref{eq:energy})$, we consider two equations of state (cf. \cite[(1.5.14)]{Batchelor_IntroductiontoFluid_2002})
\begin{align*}
dv&=\alpha_p v \, dT -\beta_T v \, dp, \\
ds&=\frac{c_p}{T} \,dT - \alpha_p v \,dp, \\
\end{align*}
where 
\begin{equation*}
\alpha_p=\frac{1}{v}\left( \frac{\partial v}{\partial T} \right)_p,
\beta_T=-\frac{1}{v}\left( \frac{\partial v}{\partial p} \right)_T,
c_p=T\left( \frac{\partial s}{\partial T} \right)_p,
v=\frac{1}{\rho},
\end{equation*}
denotes the volume expansivity, isothermal compressibility, specific heat at constant pressure and the specific volume, respectively. 
If we can express a thermodynamic property as a function of two other variables of state, the fluid is called divariant.
Consequently, the thermodynamic state of a general divariant fluid can be determined by the three measurable quantities $\alpha_p, \beta_T$ and $c_p.$
Moreover, we can write equivalent equations of state under consideration of the speed of sound
$c=\left( \partial p/\partial \rho \right)_s^{1/2}.$
To do so, using
(cf. \cite{Chalot.Hughes.ea_Symmetrizationof_2002})
\begin{equation*}
c_v=c_p-\frac{\alpha_p^2 v T}{\beta_T}, 
\quad c^2=\frac{c_p}{c_v \rho \beta_T},
\end{equation*}
and interpreting $d\rho, dp$ and $ds$ along the pathlines,
we conclude
\begin{align*}
\frac{D\rho}{Dt}&=\frac{1}{c^2} \frac{Dp}{Dt} - \frac{\rho \alpha_p T}{c_p} \frac{Ds}{Dt},\\
\frac{DT}{Dt}&=\frac{T \alpha_p}{\rho c_p} \frac{Dp}{Dt}+ \frac{T}{c_p} \frac{Ds}{Dt}.
\end{align*}
In the following we restrict our attention to isentropic flows with $\frac{Dc}{Dt}=0.$ As a consequence, the energy equation $(\ref{eq:energy})$ reduces to $\frac{Ds}{Dt}=0.$ Thus we obtain the closed set of equations
\begin{align*}
\frac{D \rho}{Dt} +\rho \nabla \cdot u&=0 \quad \text{in } (0,\tau) \times \R^d, 
\\
\rho \frac{D u}{Dt} +\nabla p &=\nabla \cdot \sigma \quad \text{in }  (0,\tau) \times \R^d, 
\\
\frac{D\rho}{Dt}&=\frac{1}{c^2} \frac{Dp}{Dt} \quad \text{in }  (0,\tau) \times \R^d,
\\
(p,u^\top,\rho)^\top&=g \text{ in }\left\{ 0 \right\} \times \R^d.
\end{align*}
The dimensionless form of the isentropic Navier-Stokes equation is obtained  
by using the following non-dimensional variables.
For simplicity of notation, we replace 
\begin{equation*}
\begin{matrix}
\frac{t c_R}{x_R} &\to& t,\quad  & \frac{x}{x_R} &\to& x,\quad & \frac{\rho}{\rho_R} &\to& \rho, \\
\frac{u}{c_R} &\to& u,\quad  & \frac{p}{\rho_R c_R^2} &\to& p,\quad & \frac{\mu}{\mu_R} &\to& \mu, \\
\frac{c}{c_R} &\to& c,\quad  & \frac{x_R}{c_R} \partial_t &\to& \partial_t,\quad & x_R \nabla &\to& \nabla, \\
\end{matrix}
\end{equation*} 
where variables with index $R$ denote some reference values.
Consequently, it follows that 
\begin{align}
\frac{D \rho}{Dt} +\rho \nabla \cdot u&=0 \quad \text{in } (0,\tau) \times \R^d, \label{eq:continuity_isentrop} \\
\rho \frac{D u}{Dt} +\nabla p &=\frac{\Ma}{\re}\nabla \cdot \sigma \quad \text{in }  (0,\tau) \times \R^d, \label{eq:momentum_isentrop} \\
\frac{D\rho}{Dt}&=\frac{1}{c^2} \frac{Dp}{Dt} \quad \text{in }  (0,\tau) \times \R^d,\label{eq:state1_isentrop} \\
(p,u^\top,\rho)^\top&=g \text{ in }\left\{ 0 \right\} \times \R^d,
\label{eq:initial_cond_isentrop}
\end{align}
where $\Ma=\frac{u_R}{c_R}$ and $\re=\frac{\rho_R u_R x_R}{\mu_R}$ are the global Mach and Reynolds number, respectively.

\section{Conservation laws in the presence of uniform mean flow}
In this section we give a brief introduction on conservation laws which can
be considered as the most fundamental principles in physics.
Using standard arguments, it is shown that the fluctuations around an uniform mean flow can be formulated as a conservation law.

Let $U$ be an intrinsic quantity of a principle under consideration and $F_i(U),i=1,\dots,d$ a vector field such that the physical system satisfies (with the help of Einstein's summation convention) 
\begin{equation*}
\partial_t U + \partial_{x_i} F_i(U)=0. 
\end{equation*}
Integrating over an arbitrary but fixed domain $\Omega$ and applying the divergence Theorem, we see that the rate of change of the quantity $U$ inside $\Omega$ is equal to the flux across the boundary $\partial \Omega$
\begin{equation*}
\frac{d}{dt} \int_{\Omega} U \, dx=-\int_{\partial \Omega} F(U)\cdot n \, ds,
\end{equation*}
where $n$ denotes the unit outward normal vector. 
If the equations also include dissipation effects the conservations law is called balance law.
For instance, 
the compressible Navier-Stokes equations of the previous section 
are balance laws which 
involves acoustic effects as well as other compressible pressure fluctuations. 
Assuming the simplest case where the medium is at rest ($\bar{u}=0$) 
and carrying out the steps 1.-4. 
for the isentropic fluctuations 
$(p'=c^2 \rho')$
we satisfy
the conservation law by setting
\begin{equation*}
U
=
\begin{pmatrix}
\frac{p'}{\bar{\rho}c^2} \\
\bar{\rho} u_1' \\
\vdots \\
\bar{\rho} u_d' \\
\end{pmatrix}, \quad
F_1(U)
=
\begin{pmatrix}
\frac{U_2}{\bar{\rho}} \\
\bar{\rho} c^2 U_1\\
0\\
\vdots\\
0\\
\end{pmatrix}, \dots , \quad
F_{d}(U)
=
\begin{pmatrix}
\frac{U_{d+1}}{\bar{\rho}} \\
0\\
\vdots\\
0,\\
\bar{\rho} c^2 U_1
\end{pmatrix},
\end{equation*}
where $(\bar{\cdot})$ denotes the time mean value.
Next, we can eliminate the velocity fluctuations and
deduce 
the wave equation  
\begin{equation}
\partial_t^2 \left( \frac{p'}{c^2 \bar{\rho}} \right) - \nabla \cdot \left( \frac{1}{\bar{\rho}} \nabla p' \right)=0.
\label{eq:wave}
\end{equation}
Now, we consider a moving medium with mean flow $(\bar{u} = \const)$ 
together with 
a moving frame with coordinates $X \in \R^d$
and
a fixed reference frame with coordinates 
$x(t)=X+\bar{u} t.$ 
Since the medium is at rest in the moving frame, it follows that 
the wave equation 
\begin{equation}
\partial_t^2 \left( \frac{P'}{c^2 \bar{\rho}} \right) - \nabla_{X} \cdot \left( \frac{1}{\bar{\rho}} \nabla_{X} P' \right)=0
\label{eq:wave_moving}
\end{equation}
is still valid for $P'(t,X)=p'(t,x(t,X)).$
In order to perform the transformation of $(\ref{eq:wave_moving})$ into the fixed frame, we differentiate $F(t,X)=f(t,x(t,X)).$ That is, we have 
\begin{equation}
\partial_t F=\partial_t f+ \bar{u} \cdot \nabla f, \quad \partial_{X_i} F=\partial_{x_j}f \, \partial_{X_i} x_j=\partial_{x_i} f.
\label{eq:Lagrange_Euler}
\end{equation} 
Introducing the notation $\frac{\bar{D}}{Dt}=\partial_t + \bar{u} \cdot \nabla$
we can write the convective wave equation
\begin{equation}
\frac{\bar{D}^2}{Dt} \left(\frac{p'}{c^2 \bar{\rho}} \right)  - \nabla \cdot \left( \frac{1}{\bar{\rho}} \nabla p' \right) =0.
\label{eq:conv_wave}
\end{equation}
Since $\partial_{X_i} x_j$ is not a constant diagonal matrix in general, 
$(\ref{eq:Lagrange_Euler})$ shows that one cannot generalise this strategy. 

On the other hand,
$(\ref{eq:conv_wave})$
follows from a conservation law with
\begin{equation*}
F_1(U)
=
\begin{pmatrix}
\bar{u}_1 U_1& + &\frac{U_2}{\bar{\rho}} \\
\bar{u}_1 U_2& + &\bar{\rho} c^2 U_1\\
\bar{u}_1 U_3 & &\\
\vdots\\
\bar{u}_1 U_{d+1} & &\\
\end{pmatrix}, \dots , \quad
F_{d}(U)
=
\begin{pmatrix}
\bar{u}_d U_1 &+&\frac{U_{d+1}}{\bar{\rho}} \\
\bar{u}_d U_2 & & \\
\vdots\\
\bar{u}_d U_d & & \\
\bar{u}_d U_{d+1} &+& \bar{\rho} c^2 U_1
\end{pmatrix},
\end{equation*}

which we can derive 
directly from the compressible Navier-Stokes equations, if the mass conservation for the mean flow $\bar{u} \cdot \nabla \bar{\rho}=- \bar{\rho} \nabla \cdot \bar{u}=0$ is considered.
Notice, that for $\bar{\rho}=\const$ plane waves satisfy the impedance relation $p'/u'\cdot n=\bar{\rho}c.$  
Hence, together with $p'=c^2 \rho'$ we have
\begin{equation*}
\frac{|\rho'|}{|\bar{\rho}|}=\frac{|u' \cdot n|}{c} \le \frac{\|u'\|}{c}=\Ma'.
\end{equation*}
Thus, for general solutions of the wave equation the fluctuation mach number $\Ma'$ is a measure for the error in the linearisation step. 
Let us finally point out, that the differential operators of the conservation laws 
$(\ref{eq:wave})$ and $(\ref{eq:conv_wave})$
are also the differential operators of the analogies of Lighthill \cite{Lighthill_soundgeneratedaerodynamically._1952} and 
Möhring \cite{Moehring_wellposedacoustic_1999}.

\section{Definition of the base flow}
In order to consider the base flow, 
we recall some properties of time averaging. Let $f:[0,\tau) \to \R$ be a given function. Then, the time average of $f$ may be interpreted as the solution of the following ordinary differential equation
\begin{equation}
\begin{split}
d_t f_0 &= 0, \quad \text{in } (0,\tau),\\
f_0 &= \frac{1}{\tau} \int_{0}^{\tau} f(t) \, dt, \quad \text{in } \left\{ 0 \right\}.
\end{split}
\label{eq:time_av_base_flow}
\end{equation}
Moreover, 
consider
\begin{equation}
\begin{split}
d_t g &=\frac{f}{\tau}, \quad \text{in } (0,\tau),\\
g&=0, \quad \text{in } \left\{ 0 \right\}, 
\end{split}
\label{eq:time_av_integral}
\end{equation}
 and applying the fundamental Theorem of calculus 
the above integral can be 
written as $g(\tau).$
In the following we introduce the pathline-averaged base flow by time-averaging of a function which characterises a fluid particle $X.$ 
To do so, let $\stx=(x_0,x^\top)^\top$
denote a point in the time-space domain 
at position $x=(x_1,x_2,\dots,x_d)^\top$ 
and time $x_0=t.$ 
Now, for every $\stx$ and every velocity field 
$\stu=(1,u^\top)^\top:[0,\tau) \times \R^d \to \R^{d+1}$
one 
can find a fluid particle such that the movement of the particle 
$\stx_X(t)=(t,x(t)^\top)^\top$ defined by
\begin{equation}
\begin{split}
\dot{\stx}_X(t)&=\stu(\stx_X(t)),\\
\stx_X(0)&=(0,X^\top)^\top,
\label{def:pathline}
\end{split}
\end{equation}
satisfies $\stx=\stx_X(t).$
It remains to consider functions of the form $f,g: [0,t) \times \R^d \to \R.$
From the Lagrangian point of view we have to change $f_0 \to f_0 \circ \stx_X$ and $g \to g \circ \stx_X$ in $(\ref{eq:time_av_base_flow})$ and $(\ref{eq:time_av_integral})$ resp.
\begin{equation}
\begin{split}
d_t (g \circ \stx_X)\underset{(\ref{def:pathline})}{=}\partial_t g + u \cdot \nabla g  &= \frac{f}{\tau}, \quad \text{in } (0,\tau) \times \R^d,\\
g \circ \stx_X&=0, \quad \text{in } \left\{ 0 \right\} \times \R^d,
\end{split}
\label{eq:path_av_integral}
\end{equation}
and
\begin{equation}
\begin{split}
d_t (f_0 \circ \stx_X) \underset{(\ref{def:pathline})}{=} 
\partial_t f_0 + u \cdot \nabla f_0 &= 0, \quad \text{in } (0,\tau) \times \R^d,\\
f_0 \circ \stx_X &=  (g \circ \stx_X)(\tau)
, \quad \text{in } \left\{ 0 \right\} \times \R^d.
\end{split}
\label{eq:base_flow}
\end{equation}
In other words, we have defined the base flow as
\begin{equation}
f_0(t,x)=\left\{ \frac{1}{\tau} \int_{0}^{\tau} f(\stx_X(\xi)) \, d\xi: (t,x) =\stx_X(t) \right\}, 
\label{def:base_flow}
\end{equation}
which we identify in the following also as $\bar{f}(t,x).$

Note that, the implementation of the initial condition 
of $(\ref{eq:base_flow})$
requires the knowledge of
$g(\stx_X(\tau))$ at time $0$
and the knowledge of
$x(\tau),$ resp. where only $x(0)$ is available.
In order to avoid this problem, we can interpret $(\ref{eq:base_flow})$ as an ``end value problem''
\begin{equation}
\begin{split}
\partial_t f_0 + u \cdot \nabla f_0 &= 0, \quad \text{in } (0,\tau) \times \R^d,\\
f_0 \circ \stx_X &=  (g \circ \stx_X)(\tau)
, \quad \text{in } \left\{ \tau \right\} \times \R^d,
\end{split}
\label{eq:base_flow2}
\end{equation}
which is equivalent to the following backward problem for the time variable $t^-=\tau -t$ 
\begin{equation}
\begin{split}
\partial_{t^-} f_0 -u \cdot  \nabla f_0 &= 0, \quad \text{in } (t_0^-,t_0^- +\tau) \times \R^d,\\
f_0 &= g(t_0^-,x) , \quad \text{in } \left\{ t_0^- \right\} \times \R^d.
\end{split}
\label{eq:base_flow_backward}
\end{equation}

Summarising the above considerations, we have to solve $(\ref{eq:path_av_integral})$ and $(\ref{eq:base_flow_backward})$ to determine the base flow form the given flow variable $f$.

\section{Propagation of sound}
In this section we 
investigate the propagation of sound for highly-unsteady flows. 
Namely, a system of perturbation equations in conservation form and, as a consequence, a closed system of equations for the acoustic quantities are derived.
These results are consequences 
of an asymptotic analysis of the 
base flow and the isentropic Navier-Stokes equations followed by a projection onto the acoustic perturbations. 

The perturbations are the first order fluctuations around
the pathline-averaged base flow $\bar{f}(\stx_X(t)),$
which satisfies $(\ref{eq:mean_flow_cond}).$
We start
with an uniform asymptotic expansion (cf. \cite[p. 20]{Johnson_SingularPerturbationTheory_2005}) 

\begin{equation}
\begin{split}
f(\stx_X(t);\eps)
&= f_0(\stx_X(t)) + f_1(\stx_X(t))\eps + O(\eps^2)\\
&= \bar{f}(\stx_X(t)) + f'(\stx_X(t)) + O(\eps^2),
\end{split}
\label{def:fluctuating}
\end{equation} 
where $\bar{f}=f_0$ and $f_1$ are independent of $\eps.$
One important property of asymptotic expansions is that there is no need for a convergent-series representation of $f.$
Therefore, the following is assumed:
\begin{ann}
The solutions $\rho, u$ and $p$ of $(\ref{eq:continuity_isentrop})-(\ref{eq:initial_cond_isentrop})$ can be decomposed by an uniform asymptotic expansion around the pathline-averaged base flow
\begin{equation}
\begin{split}
\rho(t,x;\eps)&=\overline{\rho}(t,x)+\rho'(t,x) + O(\eps^2),\\
u(t,x;\eps)&=\overline{u}(t,x)+u'(t,x) + O(\eps^2), \\
p(t,x;\eps)&=\overline{p}(t,x)+p'(t,x) + O(\eps^2), \quad 0 < \eps \ll 1.
\end{split}
\label{eq:ass}
\end{equation}
\label{ass:one}
\end{ann}
Then, by 
substituting the decompositions $(\ref{eq:ass})$
in $(\ref{eq:continuity_isentrop})-(\ref{eq:initial_cond_isentrop}),$
we can apply 
a well known result from the asymptotic analysis which leads to a hierarchy of equations for the terms multiplied by the same power of $\eps:$ 
\begin{lemma}
Let $L_i,\, i=0,\cdots,N$ be arbitrary terms which are independent of $\eps.$ Then the statement
\begin{equation}
\sum_{i=0}^N L_i \eps^{i}= O(\eps^{N+1}), \, \eps \to 0,
\label{eq:asymp_key_result}
\end{equation}
holds if and only if $L_i=0,\, i=0,\cdots,N.$
\label{lem:asymp_key_result}
\end{lemma}

\begin{proof}
\cite[Lemma 3.1]{Meister_AsymptoticSingleand_1999}.
\end{proof}

Hence, from $\frac{Dc}{Dt}=0$ it follows that 
$(\ref{eq:state1_isentrop})$ 
can be solved by setting
\begin{equation}
\rho_i=p_i/c^2, \quad i=0,\cdots,N.
\label{eq:state1_coeff}
\end{equation}
Moreover, $(\ref{eq:continuity_isentrop})$ and $(\ref{eq:momentum_isentrop})$ reduce to
\begin{align}
\frac{D \rho'}{Dt} 
+\rho'  \nabla \cdot \bar{u} 
+\bar{\rho}  \nabla \cdot u' 
+\rho'   \nabla \cdot u' 
&=- \bar{\rho} \nabla \cdot \bar{u} +O(\eps^2),
\label{eq:perturbed_continuity}\\
\left(\overline{\rho}+\rho'\right) \frac{D u'}{Dt} 
+\nabla p' 
- \frac{\Ma}{\re}\nabla \cdot \sigma' 
&=-\nabla \bar{p} + \frac{\Ma}{\re}\nabla \cdot \bar{\sigma} +O(\eps^2). 
\label{eq:perturbed_momentum} 
\end{align}
Now, $(\ref{eq:state1_coeff})$ immediately implies
\begin{align}
\frac{D}{Dt} \frac{p'}{c^2} 
+\bar{\rho}  \nabla \cdot u' 
&=- \bar{\rho} \nabla \cdot \bar{u}
-\frac{p'}{c^2} \nabla \cdot \bar{u} 
-\frac{p'}{c^2} \nabla \cdot u' 
+O(\eps^2),
\label{eq:pfluc} \\
\overline{\rho}\frac{D u'}{Dt} 
+\nabla p' 
- \frac{\Ma}{\re}\nabla \cdot \sigma' 
&=-\nabla \bar{p} + \frac{\Ma}{\re}\nabla \cdot \bar{\sigma}
- \frac{p'}{c^2} \frac{Du'}{Dt} 
+O(\eps^2),
\label{eq:ufluc} 
\end{align}
resp.
\begin{align}
\partial_t \left( \frac{p'}{\bar{\rho}c^2} \right) 
+\nabla \cdot \left( \frac{p'}{\bar{\rho}c^2} u +u' \right) 
&=- \nabla \cdot \bar{u} +O(\eps^2),
\label{eq:pfluc_conserv} \\
\partial_t\left( \bar{\rho} u' \right) 
+ \nabla \cdot \left( \bar{\rho} u' \otimes u +p'I- \frac{\Ma}{\re}\sigma'\right) 
&=- \nabla \bar{p} + \frac{\Ma}{\re}\nabla \cdot \bar{\sigma}
-\frac{p'}{c^2} \frac{Du'}{Dt} 
+\bar{\rho} u' \nabla \cdot u 
+O(\eps^2),
\label{eq:ufluc_conserv} 
\end{align}
where we have used the notation $\nabla \cdot (a \otimes b)=\partial_i (a_j b_i),$ $i,j=1,\cdots,d.$

\begin{lemma}
Let Assumption $\ref{ass:one}$ be satisfied, 
$\frac{Dc}{Dt} =0$ and  
$\frac{\Ma}{\re} = O(\eps^k),\, k \ge 1.$
Then, for the base  
components of $(\ref{eq:ass})$
it holds that 
\begin{enumerate}
\item[(i)]
\begin{equation}
\nabla \cdot \bar{u}=0,\quad \bar{p}=\text{constant},
\label{eq:ns_asymp_first_order}
\end{equation}
\item[(ii)] 
\begin{equation}
\frac{\bar{D}}{Dt} \bar{u}=0, \quad u' \cdot \nabla \bar{u}=0,
\label{eq:pathline_asymp}
\end{equation}
\item[(iii)]
\begin{equation}
\nabla \bar{u}:\nabla \bar{u}=\partial_i \bar{u}_j \partial_j \bar{u}_i=0.
\label{eq:pathline_asymp2}
\end{equation}
\end{enumerate}
\label{lem:baseflow}
\end{lemma}

\begin{proof}
In the following we are 
using Lemma \ref{lem:asymp_key_result} again.
Due to the fact that $\rho_i,u_i,p_i,\, i=0,1$ are independent of $\eps,$  
the first assertion follows by the  
consideration of leading order terms 
in $(\ref{eq:pfluc_conserv})$ and $(\ref{eq:ufluc_conserv})$ 
and the definition of the pathline-averaged pressure. 
Next, the definition of the pathline-averaged velocity implies the second statement. 
Finally, the identity
\begin{equation*}
\nabla \cdot \frac{\bar{D}}{Dt}{\bar{u}}=\frac{\bar{D}}{Dt} \nabla \cdot \bar{u}+ \nabla \bar{u}:\nabla \bar{u}
\end{equation*} 
completes the proof.
\end{proof}
\begin{lemma}
Let Assumption $\ref{ass:one}$ be satisfied, 
$c=\text{constant}$ and  
$\frac{\Ma}{\re} = O(\eps^k),\, k \ge 1.$
Then, for the base  
components of $(\ref{eq:ass})$
it holds that 
\begin{equation}
\bar{\rho}=\frac{\bar{p}}{c^2}=\text{constant}.
\label{eq:rho_bar}
\end{equation}
\label{lem:baseflow2}
\end{lemma}
\begin{proof}
The desired result follows by using $(\ref{eq:state1_coeff})$ and Lemma $\ref{lem:baseflow}.$
\end{proof}
Note that for the aforementioned Lemmas 
no assumption on a slowly varying flow is necessary.
Furthermore, the approach under consideration can be interpreted as a perturbation ansatz about an incompressible base flow.

Now, we obtain a conservation law for the perturbations around the base flow.  

\begin{theorem}
Let Assumption $\ref{ass:one},$ 
$\frac{Dc}{Dt} =0$ and  
$\frac{\Ma}{\re} = O(\eps^k),\, k \ge 2$ be satisfied
and let $g:\R^d \to \R \times \R^d$ be a known function.
Then, the first order fluctuation 
components of $(\ref{eq:ass})$
are solutions of 
\begin{align}
\partial_t \left( \frac{p'}{\bar{\rho}c^2} \right) 
+\nabla \cdot \left( \frac{p'}{\bar{\rho}c^2} \bar{u} +u' \right)
&=0 
\quad \text{in } (0,\tau) \times \R^d,
\label{eq:pfluc_conserv_neg2} \\
\partial_t\left( \bar{\rho} u' \right)+ \nabla \cdot \left( \bar{\rho} u' \otimes \bar{u} +p'I\right)
&= 0
\quad \text{in } (0,\tau) \times \R^d, 
\label{eq:ufluc_conserv_neg2} \\
(p',u'^\top)^\top&=g \text{ in }\left\{ 0 \right\} \times \R^d.
\label{eq:ini_fluc_conserv_neg2} 
\end{align}
\label{th:conserv_law}
\end{theorem}
\begin{proof}
The proof is 
obtained by using Lemma \ref{lem:baseflow} and the equations  
for the first order terms. 
\end{proof}
\begin{remark}
The assumption  
$\frac{\Ma}{\re} = O(\eps^k),\, k\ge 2$ may be changed to  
$\frac{\Ma}{\re} = O(\eps)$ by adding $\frac{\Ma}{\re}\nabla \cdot \bar{\sigma}$ on the right hand side of 
$(\ref{eq:ufluc_conserv_neg2}).$ 
Moreover, $\frac{\Ma}{\re}$ can be interpreted as a linearisation error. 
\end{remark}

\begin{theorem}[Splitting Theorem]
Let Assumption $\ref{ass:one},$
$c=\text{constant}$ and  
$\frac{\Ma}{\re} = O(\eps^k), \, k \ge2$ be satisfied.
Then, the first order fluctuation 
components of $(\ref{eq:ass})$ can be split uniquely into an acoustic variable,
which is radiated by a $c-$dependent velocity, and a hydrodynamic variable, i.e.
\begin{equation}
\begin{pmatrix}
p' \\
u'
\end{pmatrix}
=
\begin{pmatrix}
p' \\
u^a
\end{pmatrix}
+
\begin{pmatrix}
0 \\
u^v
\end{pmatrix},
\quad \nabla \times u^a=0, \quad \nabla \cdot u^v=0.
\label{eq:splitting}
\end{equation}
Moreover, the acoustic variables satisfy independent equations
\begin{align}
\frac{\bar{D}}{Dt} 
\left( \frac{p'}{\bar{\rho}c^2} \right) 
+\nabla \cdot u^a
&=0, 
\label{eq:proj_p_a} \\
\frac{\bar{D}}{Dt} \nabla \cdot u^a
+\frac{1}{\bar{\rho}} \Delta p' 
&= 0,
\label{eq:proj_u_a} \\
\nabla \times \frac{\bar{D}}{Dt} u^v &=-\nabla \times \nabla \cdot ( u^a \otimes \bar{u}).
\label{eq:proj_u_v} 
\end{align} 
\label{th:splitting}
\end{theorem}

\begin{proof}
Using the notation
$U=(p',u'^\top)^\top$ 
and denote the vectors of unity by $e_j,\, j=1,\cdots,d+1,$
the equations 
$(\ref{eq:pfluc_conserv_neg2})$ and
$(\ref{eq:ufluc_conserv_neg2})$ directly give
\begin{equation*}
L(U)=\partial_t U +A_j \partial_{x_j} U =-\partial_{x_j}s_j,
\end{equation*}
where 
\begin{equation*}
A_j 
= \bar{\rho} c^2 e_1 e_{j+1}^\top +  \frac{1}{\bar{\rho}} e_{j+1} e_1^\top,\notag 
\end{equation*}
are constant matrices
and 
\begin{equation*}
s_j
= (p' \bar{u}_j ) e_1
+\sum_{k=1}^{d} \underbrace{u_k' \bar{u}_j}_{=S_{kj}} e_{k+1}.  
\end{equation*}
Therefore, we can apply the Laplace / Fourier transformation 
(cf. \cite[Appendix A]{Ewert.Schroeder_Acousticperturbationequations_2003})
$\hat{f}(\bar{\omega},\alpha)=\mathcal{LF}[f(t,x)](\bar{\omega},\alpha),$ $\bar{\omega}=-\omega +i \sigma, \, \alpha \in \R^d,$
that is 
\begin{equation*}
\hat{f}(\bar{\omega},\alpha)
=\frac{1}{(2 \pi)^{d+1}}\int_0^\infty \int_{\R^d} f(t,x) e^{-i(\alpha \cdot x- \bar{\omega}t)} \, dx \,dt.
\end{equation*}
The inverse transformation is defined by
\begin{equation*}
f(t,x)=\mathcal{LF}^{-1}[\hat{f}(\bar{\omega},\alpha)](t,x)
=\int_{-\infty+i \sigma}^{\infty +i \sigma} \int_{R^d}
 \hat{f}(\bar{\omega},\alpha) e^{i(\alpha \cdot x-\bar{\omega}t)} \, d\alpha \, d\bar{\omega},\quad \sigma > \sigma_0,
\end{equation*}
where $\sigma$ is a constant number such that the path of integration is in the region of convergence of $\hat{f}(\bar{\omega},\alpha).$
Using the following properties of the Laplace / Fourier transformation
\begin{equation*}
\begin{split}
\mathcal{LF}[\lambda f+ \mu g]&= \lambda \, \mathcal{LF}[f]+\mu \,\mathcal{LF}[g], \quad \lambda, \mu \in \R, \\
\mathcal{LF}[\partial_t f]&= -i \bar{\omega}\mathcal{LF}[f]
-\frac{1}{2 \pi} f_0, \\
\mathcal{LF}[\partial_{x_j} f]&=i \alpha_j \mathcal{LF}[f],
\quad j=1,\cdots,d, \\
\end{split}
\end{equation*}
where 
\begin{equation*}
f_0(\alpha)
=\frac{1}{(2 \pi)^d} \int_{\R^d} f(0,x)e^{-i\alpha \cdot x} \, dx,
\end{equation*}
it immediately implies 
\begin{equation*}
\mathcal{LF}[L(U)]
=-i \underbrace{ \left(\bar{\omega}I-\alpha_j A_j \right) }_{=A}   \hat{U}
- \frac{1}{2 \pi} U_0
=-i \alpha_j \hat{s}_j.
\end{equation*}
In order to separate the $c-$dependent acoustic terms from the hydrodynamic terms, we consider the decomposition $A=R \Lambda R^{-1},$ with the matrices given by
\begin{equation*}
R=
\begin{pmatrix}
1 & 1 & 0 & 0 & \cdots & 0 \\
\multirow{5}{*}{$\displaystyle \frac{\alpha}{\| \alpha \|c \bar{\rho}}$} &
\multirow{5}{*}{$\displaystyle \frac{-\alpha}{\| \alpha \|c \bar{\rho}}$} 
 & \alpha_d &  0 &\cdots & 0 \\
 & & 0  & \alpha_d  & & \vdots \\
 & & \vdots  &   &\ddots & 0 \\
 & & 0  & \cdots  &0 & \alpha_d \\
 & & -\alpha_1 & -\alpha_2 & \cdots & -\alpha_{d-1} \\
\end{pmatrix}
\end{equation*}
and
\begin{equation*} 
\Lambda =
\begin{pmatrix}
-\| \alpha \|c +\bar{\omega} & 0 &  0&  0 &\cdots & 0 \\
0 & \| \alpha \|c+\bar{\omega}     & 0 &0 &\cdots & 0 \\
0 &   0  & \bar{\omega}  & 0& \cdots & 0 \\
0 &   0  & 0  & \bar{\omega }& & \vdots \\
\vdots &   \vdots  & &  &\ddots & 0 \\
0 &   0  & 0&\cdots  &0 & \bar{\omega} \\
\end{pmatrix}.
\end{equation*}
Then, the decoupled system of equations follow as
\begin{equation*}
\Lambda R^{-1} \hat{U}=i R^{-1}\left(- i \alpha_j \hat{s}_j+\frac{1}{2 \pi} U_0 \right).
\end{equation*}
Now, by the concrete structure of $R$ and $\Lambda,$ we obtain the projected acoustic equations
\begin{equation}
R I_a\Lambda R^{-1} \hat{U}= R I_a R^{-1}\left( \alpha_j \hat{s}_j+\frac{i}{2 \pi} U_0 \right)
\label{eq:four_a}
\end{equation}
and the projected hydrodynamic equations
\begin{equation}
R I_v\Lambda R^{-1} \hat{U}= R I_v R^{-1}\left(\alpha_j \hat{s}_j+\frac{i}{2 \pi} U_0 \right),
\label{eq:four_v}
\end{equation}
where $I_a=e_1 e_1^\top + e_2 e_2^\top$ and $I_v=I-I_a.$
Using the notation 
$B=(\alpha \alpha^\top)/\|\alpha\|^2,$
it is not difficult to see that 
\begin{equation*}
R I_a \Lambda R^{-1}=
\begin{pmatrix}
\bar{\omega} & -\alpha_1 c^2 \bar{\rho} & \cdots &-\alpha_d c^2 \bar{\rho} \\
- \frac{\alpha_1}{\bar{\rho}} & & & \\
\vdots  & & \bar{\omega} B  & \\
- \frac{\alpha_d}{\bar{\rho}} & & & \\
\end{pmatrix}
\end{equation*}
and
\begin{equation*}
R I_a R^{-1}=
\begin{pmatrix}
1 & 0 & \cdots  & 0 \\
0 & & & \\
\vdots  & & B  & \\
0 & & & \\
\end{pmatrix}
\end{equation*}
are equal in the first row to $A$ and $I,$ resp. 
Therefore, we conclude that $(\ref{eq:pfluc_conserv_neg2})$ is a purely acoustic equation. 
The remaining acoustic equations lead to 
\begin{equation*}
-\frac{\alpha}{\bar{\rho}} \hat{p}' + \bar{\omega} B \hat{u}'
=B \left(\hat{S} \alpha+\frac{i}{2\pi} u_{0} \right), 
\end{equation*}
which reduce to the scalar equation
\begin{equation*}
-\frac{\|\alpha\|^2}{\bar{\rho}} \hat{p}' + \bar{\omega} \alpha^\top \hat{u}'
=\alpha^\top \hat{S} \alpha + \frac{i }{2 \pi} \alpha^\top u_0.
\end{equation*}

Now, the back transformation yields
\begin{equation}
\partial_t \nabla \cdot u' +\frac{1}{\bar{\rho}} \Delta p'
+ \nabla \cdot \nabla \cdot \left( u' \otimes \bar{u} \right)=0. 
\end{equation}
Applying the product rule, it follows that 
$\nabla \cdot \nabla \cdot \left( \bar{u} \otimes u' \right)
=\nabla \cdot \nabla \cdot \left( u'\otimes \bar{u} \right).$
Moreover, we obtain from $(\ref{eq:pathline_asymp})$
\begin{equation*}
\nabla \cdot (\bar{u} \otimes u')=\bar{u} \nabla \cdot u',
\end{equation*}
which leads to
\begin{equation}
\partial_t \nabla \cdot u' +\frac{1}{\bar{\rho}} \Delta p'
+ \nabla \cdot \left( \bar{u} \nabla \cdot u' \right) 
= 0.
\end{equation}
By the Helmholtz decomposition (cf. \cite[Chapter 3.5]{Achenbach_Wavepropagationin_1973})
we can decouple  
the 
irrotational and solenoidal 
parts of the velocity perturbation such that
\begin{equation*}
u'=u^{a} +u^{v} =\nabla \varphi + \nabla \times \psi,\quad \nabla \cdot \psi=0.
\label{eq:helmholtz_decomp1}
\end{equation*}
Consequently, 
\begin{equation*}
\nabla \times u^{a} =0,\quad \nabla \cdot u^{v}=0, 
\label{eq:helmholtz_decomp1}
\end{equation*}
proves the  assertion $(\ref{eq:proj_u_a}).$ 

Second, concerning the hydrodynamic projection we have
\begin{equation*}
R I_v \Lambda R^{-1}=
\begin{pmatrix}
0 &  \cdots & 0 \\
\vdots  & \, \bar{\omega} C  & \\
0 & & \\
\end{pmatrix}, \quad
R I_v R^{-1}=
\begin{pmatrix}
0 &  \cdots  & 0 \\
\vdots  & \, C  & \\
0 & & \\
\end{pmatrix}, \quad
C=\frac{\left( \| \alpha \|^2 I - \alpha \alpha^\top \right)}{\| \alpha \|^2}.
\end{equation*}
Thus, we obtain
\begin{equation}
\bar{\omega} \left( \alpha^\top \alpha I- \alpha \alpha^\top \right) \hat{u}'
=\left( \alpha^\top \alpha I- \alpha \alpha^\top \right)(\hat{S} \alpha +\frac{i}{2\pi} u_0).
\label{eq:proj_u_v2}
\end{equation}
In order to consider the back-transformation we separate the two cases 
\begin{equation*}
\begin{split}
d=2 &: \quad
\alpha^\top \alpha I - \alpha \alpha^\top
= \begin{pmatrix}
\alpha_2^2 & -\alpha_1 \alpha_2 \\
-\alpha_1 \alpha_2 & \alpha_1^2
\end{pmatrix}
=\begin{pmatrix}
\alpha_2 \\
-\alpha_1
\end{pmatrix}
\begin{pmatrix}
\alpha_2,-\alpha_1
\end{pmatrix}, \\
d=3 &: \quad
\alpha^\top \alpha I - \alpha \alpha^\top
= \begin{pmatrix}
\alpha_2^2+\alpha_3^2 & -\alpha_1 \alpha_2 & -\alpha_1 \alpha_3 \\
-\alpha_1 \alpha_2 & \alpha_1^2 + \alpha_3^2 & -\alpha_2 \alpha_3 \\
-\alpha_1 \alpha_3 & -\alpha_2 \alpha_3 & \alpha_1^2+ \alpha_2^2
\end{pmatrix}
=- 
\begin{pmatrix}
0 & -\alpha_3 &\alpha_2 \\
\alpha_3 & 0 &-\alpha_1 \\
-\alpha_2 & \alpha_1 & 0
\end{pmatrix}^2.
\end{split}
\end{equation*}
Hence, applying the Helmholtz decomposition, we conclude that the Laplace / Fourier transformation of 
\begin{equation*}
\nabla \times \frac{\bar{D}}{Dt} u^v =-\nabla \times \frac{\bar{D}}{Dt} u^a.
\end{equation*}
satisfies $(\ref{eq:proj_u_v2}).$ Thus the Splitting Theorem is proved.
\end{proof}

\begin{remark}
The definition of the base flow and Theorem \ref{th:splitting}  
answer the question raised in the introduction: 
The highly-unsteady flow variables $p$ and $u$ could be decomposed into acoustic variables $p', u^a$ and hydrodynamic variables $\bar{p}, \bar{u}+u^v,$ where the acoustic quantities satisfy the closed equations  $(\ref{eq:proj_p_a})$ and $(\ref{eq:proj_u_a}).$
Let us highlight the key ingredients:

\begin{itemize}
\item
The pathline-averaged base flow and their properties from Lemma \ref{lem:baseflow} and \ref{lem:baseflow2};
\item
The divergence structure of conservation law $(\ref{eq:pfluc_conserv_neg2})-(\ref{eq:ini_fluc_conserv_neg2});$
\item
The separation of the divergence free velocity perturbations from the irrotational part.
\end{itemize}
Notice that, the divergence structure of the conservation law was essential to formulate an equation without the divergence free velocity (see $(\ref{eq:proj_u_a})$).

\end{remark}

Moreover, the 
propagation of sound can also be formulated as follows:

\begin{corollary}
Under the assumptions of Theorem \ref{th:splitting}, the acoustic perturbations satisfy 
\begin{align}
\frac{1}{c^2} \left(\frac{\bar{D}}{Dt}\right)^2 p' - \Delta p'&=
0
\quad \text{in } (0,\tau) \times \R^d,
\label{eq:conv_wave2} \\
p'&=p'_0 \text{ in }\left\{ 0 \right\} \times \R^d.
\label{eq:ini_conv_wave2} 
\end{align}
\end{corollary}

\begin{proof}
Taking $\frac{\bar{D}}{Dt}(\ref{eq:proj_p_a})-(\ref{eq:proj_u_a})$
immediately implies the result.
\end{proof}

\section{Sources of sound}
In the previous section it was demonstrated that the propagation of sound waves are determined by the homogeneous convective wave equation. Now, we define the sources of sound as the deviations between the convective wave and the Navier-Stokes equation. To do so, 
the connection of the acoustic fluctuation along the pathline-averaged flow to the Lighthill approach is discussed.  

Therefore, the current understanding of Lighthill's solution $\rho_L'$ and $p_L',$ resp. is repeated \cite[p. 746]{Fedorchenko_Onsomefundamental_2000}: 
\begin{quote}
There is no reason for regarding $\rho_L'$ and $p_L'$ as the acoustic components because no accurate definition of acoustic components in a highly-unsteady flow has been given in reference \cite{Lighthill_soundgeneratedaerodynamically._1952}, so that $\rho_L'$ and $p_L'$ represent merely the differences between the local values of flow variables $\rho,p$ and arbitrary constants $\rho_0,p_0$.
\end{quote}
In this context let us highlight 
that the acoustic pressure perturbation satisfies Lighthill's analogy:

\begin{corollary}
Let the assumptions of Theorem \ref{th:splitting} be satisfied. 
Then, the acoustic perturbations satisfy 
\begin{align}
\frac{1}{c^2} \partial_t^2 p' - \Delta p'&= \nabla \cdot \nabla \cdot \left( \rho u \otimes u \right)+ O(\eps^2)
\quad \text{in } (0,\tau) \times \R^d,
\label{eq:lighthill} \\
p'&=p'_0 \text{ in }\left\{ 0 \right\} \times \R^d.
\label{eq:ini_lighthill} 
\end{align}
\label{cor:lighthill}
\end{corollary}

\begin{proof}
As usual, the Navier-Stokes equations $(\ref{eq:continuity_isentrop})-(\ref{eq:initial_cond_isentrop})$ imply
\begin{equation}
\partial_t^2 \rho -c^2 \Delta \rho= \nabla \cdot \nabla \cdot \left( \rho u \otimes u + \left( p - c^2 \rho \right) I - \frac{\Ma}{\re}\sigma\right).
\label{eq:first_step_lighthill}
\end{equation}
Using $(\ref{eq:state1_coeff})$ and 
$\frac{\Ma}{\re} \in O(\eps^2)$
we obtain
\begin{equation*}
\frac{1}{c^2} \partial_t^2 p - \Delta p= \nabla \cdot \nabla \cdot \left( \rho u \otimes u \right)+ O(\eps^2).
\end{equation*}
Furthermore, applying 
the decomposition $(\ref{eq:ass})$ 
and $(\ref{eq:ns_asymp_first_order})$
it follows that 
\begin{equation*}
\frac{1}{c^2} \partial_t^2 p' - \Delta p'= \nabla \cdot \nabla \cdot \left( \rho u \otimes u \right)+ O(\eps^2).
\end{equation*}
By the Splitting Theorem we conclude that $p'$ is an acoustic quantity.
\end{proof}
It is a well known fact, that the right-hand side depends on the solution variable $p'.$ Therefore, the right-hand side has to be interpreted / modelled as an equivalent source. This interpretation was often criticised because both sides are differential expressions of equal standing. 
Now, the previous results can be used to attack this problem and identify the true sources of sound.

\begin{theorem}
Let the assumptions of Theorem \ref{th:splitting} be satisfied. 
Then, the acoustic perturbations satisfy 
\begin{align}
\frac{1}{c^2} \left(\frac{\bar{D}}{Dt} \right)^2 p' - \Delta p'&= \nabla \cdot \nabla \cdot \left( \bar{\rho} u' \otimes u' \right)+ O(\eps^2)
\quad \text{in } (0,\tau) \times \R^d,
\label{eq:conv_lighthill} \\
p'&=p'_0 \text{ in }\left\{ 0 \right\} \times \R^d.
\label{eq:ini_conv_lighthill} 
\end{align}
\label{th:soundsources}
\end{theorem}

\begin{proof}
By the definition of the material derivative
and
$(\ref{eq:ns_asymp_first_order}),$ we have
\begin{equation*}
\left( \frac{\bar{D}}{Dt} \right)^2 \rho
=
\partial_t^2 \rho + \partial_t \bar{u} \cdot \nabla \rho
+2 \bar{u} \cdot \nabla \partial_t \rho 
+ \nabla \cdot \nabla \cdot \left( \rho \bar{u} \otimes \bar{u} \right)
- \nabla \bar{u}: \nabla \left( \rho \bar{u} \right)^\top.
\end{equation*} 
Then, taking into consideration $(\ref{eq:pathline_asymp2})$
and $(\ref{eq:continuity}),$ it follows
\begin{equation*}
- \nabla \bar{u}: \nabla \left( \rho \bar{u} \right)^\top
+2 \bar{u} \cdot \nabla \partial_t \rho 
=
-\left(\bar{u} \cdot \nabla \bar{u}\right)\cdot \nabla \rho
-2 \bar{u} \cdot \nabla \nabla \cdot (\rho u).
\end{equation*} 
Moreover, the relations
$(\ref{eq:ns_asymp_first_order})$
and
$(\ref{eq:pathline_asymp})$ yield
\begin{equation*}
- \bar{u} \cdot \nabla \nabla \cdot (\rho u)
=
-  \nabla \cdot \nabla \cdot \left( \rho \bar{u} \otimes u \right)
+  \nabla \cdot \left( \rho \bar{u} \cdot \nabla \bar{u} \right).
\end{equation*} 
Hence, applying
$(\ref{eq:ns_asymp_first_order})$
and
$(\ref{eq:pathline_asymp})$
again, we have
\begin{equation*}
\left( \frac{\bar{D}}{Dt} \right)^2 \rho
=
\partial_t^2 \rho - \nabla \cdot \nabla \cdot \left( \rho u \otimes u - \rho u' \otimes u'\right),
\end{equation*}
which completes the proof by 
using the analog arguments as in the proof of Corollary \ref{cor:lighthill}.
\end{proof}

From the point of view of hybrid methods, 
we have a known hydrodynamic solution $u$ and after the computation of the pathline-averaged base flow the right-hand side is known. Hence, $(\ref{eq:conv_lighthill})$ is a closed equation with the correct propagation of sound and a non-linear sound source of second order. 
In contrast to Lighthill's equation $(\ref{eq:lighthill}),$ we need an additional effort for the base flow and right hand side computation. More precisely, we have to solve $2 \times d$ linear scalar transport equations. Here, one has to decide for every application whether a known right hand side (or equivalent a well defined acoustic quantity) is needed or if an approximate right hand side is sufficient.

\section{Hyperbolic conservation laws}
In the following, it is demonstrated that the first order perturbations of $(\ref{eq:ass})$ satisfy an energy conservation law. 
This energy identity which is based on a symmetry argument
is an analogous result to the energy Theorem of Möhring's analogy \cite{Moehring_wellposedacoustic_1999}.
As a consequence, we can define the sound intensity which is a generalisation of the classical definition for highly-unsteady flows.

Notice that, the system
of Theorem \ref{th:conserv_law} 
can be written in conservation form 
\begin{align}
L(U)=\partial_t U + \partial_{x_i} F_i(U)
&=0, \text{ in } (0,\tau) \times \Omega, 
\label{eq:conservation_law}\\
U&=U_0, \text{ in } \{0\} \times \Omega.
\label{eq:conservation_law_ini}
\end{align}
Here, the solution vector and the flux functions are defined by
\begin{equation*}
U=\left( \frac{p'}{\bar{\rho} c^2},\bar{\rho} u'^{\top} \right)^{\top},\quad
F_i(U)=
\bar{u}_i U+ \frac{U_{i+1}}{\bar{\rho}}  e_1 + \bar{\rho} c^2 U_1 e_{i+1}.
\label{}
\end{equation*}
Defining the coefficient mapping $A_i : \R^{d+1} \to \R^{d+1,d+1}$ by
\begin{equation*}
A_i 
=\bar{u}_i I+\frac{1}{\bar{\rho}}e_1 e_{i+1}^\top + \bar{\rho} c^2 e_{i+1} e_1^\top,
\label{}
\end{equation*}
we obtain the identity $F_i(U)=A_i U.$
Using the product rule we find 
$\partial_{x_i} F_i(U)=A_i \partial_{x_i} U+ \partial_{x_i} A_i U.$

Now, setting
\begin{equation*}
A(\nu)= A_i \nu_i,
\label{eq:A}
\end{equation*}
for all $\nu \in \R^d,$ we can show that the system $(\ref{eq:conservation_law})$ is hyperbolic. 
In other words, the matrix $A(\nu)$ can be diagonalised as
\begin{equation*}
A(\nu)=R \Lambda R^{-1}, \quad 
\Lambda=\text{diag}(\lambda_1,\dots,\lambda_{d+1}), \quad
\lambda_1 \leq \dots \leq \lambda_{d+1} \in \R,
\label{eq:A_diag}
\end{equation*}
where the eigenvalues are given by
\begin{equation*}
\lambda_1=\bar{u} \cdot \nu- c,\, \lambda_2= \dots= 
\lambda_d=\bar{u} \cdot \nu, \lambda_{d+1}=\bar{u} \cdot \nu + c.
\label{eq:eigenvalues}
\end{equation*}

\subsection{The conservation law in symmetry variables}
Let 
$S_0:\R^{d+1}\to \R^{d+1,d+1}$ 
be a symmetric positive definite mapping 
which simultaneously symmetrises all $A_i$ from left, that is $S_0 A_i=(S_0 A_i)^\top.$
It follows that $A_i S_0^{-1}=S_0^{-1} S_0 A_i S_0^{-1}$ is also symmetric and
consequently, 
\begin{equation*}
L(U)=S_0^{-1} \partial_t (S_0 U) + A_i S_0^{-1}\partial_{x_i} (S_0 U) +  \partial_{x_i} A_i U - S_0^{-1} \partial_t S_0 U -  A_i S_0^{-1} \partial_{x_i} S_0 U. 
\end{equation*}
Moreover, setting $\hat{A}_0=S_0^{-1}, \,\hat{A}_i=A_i \hat{A}_0$ and $V=\hat{A}_0^{-1} U,$  we get 
\begin{equation}
\hat{L}(V)=L(U(V))=\hat{A}_0 \partial_t V + \hat{A}_i \partial_{x_i} V + (\partial_t \hat{A}_0  +\partial_{x_i} \hat{A}_i) V  . 
\label{eq:sym_form}
\end{equation}

Now, we consider 
$V=\left(\bar{\rho} p',\bar{\rho} u'^{\top} \right)^{\top}=\left( \bar{\rho}^2 c^2 U_1,U_2,\dots,U_{d+1} \right)^\top$
and a short calculation gives
\begin{equation*}
\hat{A}_0
=
\frac{1}{\bar{\rho}^2 c^2} e_1 e_1^\top + \sum_{j=2}^{d+1}e_j e_j^\top, \quad
\hat{A}_i
=
\frac{\bar{u}_i}{\bar{\rho}^2 c^2} e_1 e_1^\top+ \sum_{j=2}^{d+1} \bar{u}_i e_j e_j^\top+\frac{1}{\bar{\rho}} e_1 e_{i+1}^\top+\frac{1}{\bar{\rho}} e_{i+1} e_1^\top,
\label{eq:defhatA}
\end{equation*} 
which satisfies the symmetry conditions on $\hat{A}_0$ and $\hat{A}_i.$
Setting
\begin{equation*}
\hat{F}_i(V)= F_i(U(V))=\hat{A}_i V=\bar{u}_i \left( \frac{V_1}{\bar{\rho}^2 c^2},V_2,\dots,V_{d+1} \right)^\top +\frac{V_{i+1}}{\bar\rho} e_1 + \frac{V_1}{\bar{\rho}} e_{i+1},
\label{}
\end{equation*}
we have the symmetric hyperbolic system in conservation form:
\begin{align}
\hat{A}_0 \partial_t V + \partial_{x_i} \hat{F}_i(V) + \partial_t \hat{A}_0 V&=0, \text{ in } (0,\tau) \times \Omega, 
\label{eq:sym_conservation_law}\\
V&=V_0, \text{ in } \{0\} \times \Omega.
\label{eq:sym_conservation_law_ini}
\end{align}
Finally, the framework of symmetry variables for hyperbolic conservation laws allows us to prove the energy conservation law.
 
\begin{theorem}
Let the assumptions of Theorem \ref{th:splitting} be satisfied. 
Then, the 
fluctuation components of $(\ref{eq:ass})$ admit
an additional conservation law
\begin{align*}
\partial_t \eta(U) 
+\partial_{x_i} q_i(U)
&=0,
\text{ in } (0,\tau) \times \Omega, 
\\
U&=U_0, \text{ in } \{0\} \times \Omega,
\end{align*}
where
\begin{equation*}
\begin{split}
\eta(U)
&=\frac{1}{2} V(U)^\top \hat{A}_0 V(U)
=\frac{1}{2}\left( \bar{\rho}^2 c^2 U_1^2+U_2^2+\cdots+U_{d+1}^2 \right)
=\frac{1}{2} \left( \frac{p'^2}{c^2} + \bar{\rho}^2 u'^\top u' \right),\\
q_i(U)
&=\frac{1}{2}  V(U)^\top \hat{A}_i V(U) 
= \eta(U) \bar{u}_i+\bar{\rho} c^2 U_1 U_{i+1}
= \eta(U) \bar{u}_i+\bar{\rho} p' u_i'.
\end{split}
\label{eq:eta_qi}
\end{equation*}

Under the assumption that $U$ 
vanishes for large but finite $\|x\|,$ there exists the energy
\begin{equation*}
\int_{\R^d} \eta(U(t,x)) \, dx
=\int_{\R^d} \eta(U(0,x)) \, dx
, \quad  \forall t \in [0,\tau],
\label{eq:energy_identity}
\end{equation*} 
and the acoustic energy identity, resp. 
\begin{equation*}
\begin{split}
& \frac{1}{2}\int_{\R^d} \frac{p'(t,x)^2}{c^2}+ \bar{\rho}^2 (u^a(t,x))^\top u^a(t,x) \, dx \\
&=\int_{\R^d} \eta(U(0,x)) \, dx
- \frac{1}{2} \int_{\R^d} \bar{\rho}^2 (u^v(t,x))^\top u^v(t,x)\,dx
, \quad  \forall t \in [0,\tau].
\end{split}
\label{eq:energy_identity}
\end{equation*}

\label{th:energy}
\end{theorem}

\begin{proof}
By the symmetry of $(\ref{eq:sym_form})$ it follows that 
\begin{equation*}
\begin{split}
V^\top \hat{L}(V)
&= \partial_t \left( \frac{1}{2}V^\top \hat{A}_0 V\right)
+\partial_{x_i} \left(\frac{1}{2} V^\top \hat{A}_i V \right) 
+\frac{1}{2} V^\top (\partial_t \hat{A}_0
+\partial_{x_i} \hat{A}_i) V
.
\end{split}
\label{eq:VLV}
\end{equation*}
Under the conditions of Lemma \ref{lem:baseflow2} it holds that
\begin{equation*}
\begin{split}
V^\top \hat{L}(V)
&=\partial_t \left( \frac{1}{2}V^\top \hat{A}_0 V\right)
+\partial_{x_i} \left(\frac{1}{2} V^\top \hat{A}_i V \right)\\
&= \partial_t \eta(U) 
+\partial_{x_i} q_i(U)
=0. 
\end{split}
\label{}
\end{equation*}
Now, if $U$ vanishes for large but finite $\|x\|,$ then it follows that
\begin{equation*}
\int_{\R^d} V^\top \hat{L}(V)\, dx 
=\frac{d}{d\tau} \int_{\R^d}  \eta(U(\tau,x)) \, dx
=0.
\label{}
\end{equation*}
Moreover, the orthogonality 
\begin{equation*}
\int_{\R^d} (u^a)^\top u^v \, dx =- \int_{\R^d} \varphi \nabla \cdot u^v \, dx=0
\end{equation*}
finishes the proof.
\end{proof}

Finally, 
as a direct consequence of Theorem \ref{th:energy}, 
the sound intensity can be generalised for highly-unsteady flows. 
By denoting 
\begin{equation*}
\bar{f}^t(x)=\underset{\tau \to \infty}{\lim} \frac{1}{\tau} \int_{0}^{\tau} f(t,x)\,dt
\end{equation*}
the temporal average, we have
\begin{corollary}
Let the assumptions of Theorem \ref{th:splitting} be satisfied. 
Then, the 
fluctuation components of $(\ref{eq:ass})$ admit
an additional conservation law
\begin{align*}
\partial_t E(U) 
+\partial_{x_i} \left(E(U) \bar{u}_i+ p' u_i'\right)
&=0,
\text{ in } (0,\tau) \times \Omega, 
\\
U&=U_0, \text{ in } \{0\} \times \Omega,
\end{align*}
where
\begin{equation*}
\begin{split}
E(U)
&=\frac{1}{2} \left( \frac{p'^2}{\bar{\rho} c^2} + \bar{\rho} u'^\top u' \right).\\
\end{split}
\end{equation*}

Under the assumption that $E$ 
is bounded, the perturbed intensity $I=\overline{E \bar{u}+p' u'}^t$
satisfies 
$\nabla \cdot I=0.$
\end{corollary}

\section{Conclusions}
In this paper we have done a mathematical analysis of isentropic vibrations along the fluid particle path. 
It shall be outlined that in the presence of highly-unsteady flows, the paper presented the first 
decomposition of flow variables in acoustic and non-acoustic quantities,
where the acoustic quantities satisfy a closed equation (see $(\ref{eq:conv_wave2})$).
This result allowed us to give a general definition of sound and to identify the sources of sound, i.e.
\begin{equation*}
\frac{1}{c^2} \left(\frac{\bar{D}}{Dt} \right)^2 p' - \Delta p'= \nabla \cdot \nabla \cdot \left( \bar{\rho} \left(u-\bar{u}\right) \otimes \left(u-\bar{u}\right) \right).
\end{equation*}
The acoustic sources are based on a compressible velocity field $u$ and the pathline-averaged base flow $\bar{u},$ which can be obtained by the transport equations  $(\ref{eq:path_av_integral})$ and $(\ref{eq:base_flow_backward}).$
However, for low Mach number flows, sources based on an incompressible velocity field $u$ are most appropriate.

We also derived a conservation law 
for the perturbed energy, which is capable to generalise  
the sound intensity.

Moreover, the acoustic fluctuation under consideration shedded some new light on the understanding of Lighthill's wave equation.

\end{document}